\documentclass{article}

\usepackage[preprint]{neurips_2025}

\usepackage[utf8]{inputenc} 
\usepackage[T1]{fontenc}    
\usepackage{hyperref}       
\usepackage{url}            
\usepackage{booktabs}       
\usepackage{amsfonts}       
\usepackage{nicefrac}       
\usepackage{microtype}      
\usepackage{xcolor}         
 \usepackage{amsmath}
 \usepackage{graphicx}
 \usepackage{amssymb}
\usepackage{amsthm}
\usepackage{bbm}
\usepackage{tikz}
\usepackage{placeins}

\newtheorem{conj}{Conjecture}
\newtheorem{thm}[conj]{Theorem}

\newtheorem{prop}[conj]{Proposition}
\newtheorem{lemma}[conj]{Lemma}

\newtheorem{definition}{Definition}

\theoremstyle{definition}

\definecolor{applegreen}{rgb}{0.55, 0.71, 0.0}
\definecolor{alizarin}{rgb}{0.82, 0.1, 0.26}
\definecolor{slategray}{rgb}{0.44, 0.5, 0.56}
\definecolor{amber}{rgb}{1.0, 0.75, 0.0}
\definecolor{mikadoyellow}{rgb}{1.0, 0.77, 0.05}
\definecolor{cadmiumgreen}{rgb}{0.0, 0.42, 0.24}
\definecolor{forestgreen}{rgb}{0.13, 0.55, 0.13}
\definecolor{lust}{rgb}{0.9, 0.13, 0.13}
\definecolor{denim}{rgb}{0.08, 0.38, 0.74}
\definecolor{purpleheart}{rgb}{0.41, 0.21, 0.61}
\definecolor{cherryblossompink}{rgb}{1.0, 0.72, 0.77}
\definecolor{darktangerine}{rgb}{1.0, 0.66, 0.07}
\definecolor{bananayellow}{rgb}{1.0, 0.88, 0.21}
\definecolor{ltgray}{rgb}{0.83, 0.83, 0.83}

\usepackage[noend]{algpseudocode}
\usepackage{algorithm}

\newcommand{\C}{\mathcal{C}}

\newcommand{\SIIA}{\sigma_{\hbox{\rm \scriptsize IIA}}}
\newcommand{\SU}{\sigma_{\hbox{\rm \scriptsize U}}}
\newcommand{\dswap}{d_{\hbox{\rm \scriptsize swap}}}
\newcommand{\voters}{\mathcal{V}}
\newcommand{\candidates}{\mathcal{C}}

\newcommand{\profiles}{\mathcal{P}} 

\newcommand{\RR}{\mathbb{R}}
\newcommand{\TopSort}{\mathrm{TopSort}}
\newcommand{\PWCGraph}{\mathrm{PWCGraph}}

\title{Quantitative Relaxations of Arrow's Axioms}

\author{%
  Suvadip Sana\\
  Department of  Statistics and Data Science\\
  Cornell University\\
  \texttt{ss2776@cornell.edu} \\
  \And
 Daniel Brous \\
  Department of  Computer Science\\
  Cornell University\\
  \texttt{pdb62@cornell.edu} \\  
  \And
  Martin T. Wells \\
  Department of Statistics and Data Science\\
  Cornell University\\
  \texttt{mtw1@cornell.edu} 
   \And
  Moon Duchin  \\
  Department of Mathematics \\ and School of Public Policy\\
  Cornell University\\
  \texttt{mduchin@cornell.edu}}

\begin{document}

\maketitle

\begin{abstract}

In this paper we develop a novel approach to relaxing Arrow's axioms for voting rules, addressing a long-standing critique in social choice theory. Classical axioms (often styled as {\em fairness axioms} or {\em fairness criteria}) are assessed in a binary manner, so that a voting rule fails the axiom if it fails in even one corner case. Many authors have proposed a probabilistic framework to soften the axiomatic approach.  Instead of immediately passing to random preference profiles, we begin by measuring the degree to which an axiom is upheld or violated on a given profile.  We focus on two foundational axioms---Independence of Irrelevant Alternatives (IIA) and Unanimity (U)---and extend them to take values in $[0,1]$. Our $\SIIA$ measures the stability of a voting rule when candidates are removed from consideration, while $\SU$ captures the degree to which the outcome respects majority preferences. 
Together, these metrics quantify how a voting rule navigates the fundamental trade-off highlighted by Arrow's Theorem. 
We show that $\SIIA\equiv 1$ recovers classical IIA, and $\SU>0$ recovers classical Unanimity, allowing a quantitative restatement of Arrow's Theorem.  

In the empirical part of the paper, we test  these metrics on two kinds of data:  a set of over 1000 ranked choice preference profiles from Scottish local elections, and a batch of synthetic preference profiles generated with a Bradley-Terry-type model.  We use those to investigate four positional voting rules---Plurality, 2-Approval, 3-Approval, and the Borda rule---as well as the iterative rule known as Single Transferable Vote (STV).  The Borda rule consistently receives the highest $\SIIA$ and $\SU$ scores across observed and synthetic elections.  This compares interestingly with a recent result of Maskin where he shows that weakening IIA to take voter preference intensity into account can be used to uniquely pick out the Borda rule.

\end{abstract}

\section{Introduction and Motivation}\label{sec:intro}

The aggregation of individual preferences into a collective ranking or decision stands as a foundational challenge across both the classical domain of social choice theory and the modern frontiers of AI and machine learning. 
From a history centered on systems of election for political representation, the applications have expanded to recommender systems, multi-agent deliberation, and LLM-driven decision pipelines, where the need to reconcile conflicting preferences while preserving structural axioms has become increasingly consequential. Central to this discourse is Arrow's famous impossibility theorem \citep{arrow1950difficulty}, which reveals that no ranking-based voting rule can simultaneously satisfy three intuitive axioms---Independence of Irrelevant Alternatives (IIA),  Unanimity (U),  and Non-dictatorship---when three or more alternatives are present.

While the theoretical landscape around Arrow's Theorem has been extensively explored,  practical implications have remained less clear. Most studies of voting rules stick with binary criteria---either a rule satisfies IIA as a guarantee over all possible preference profiles, or it does not.\footnote{A {\em preference profile}, or simply {\em profile}, is a record of the ranking of alternatives given by each agent/voter.} This binary framing obscures meaningful variation across real-world instances and fails to offer a quantitative grasp of how badly or mildly a rule may violate an axiom. This paper offers a new direction: we propose real-valued metrics that quantify violations of the axioms on specific preference profiles, enabling empirics that allow us to compare voting rules for their degree of compliance.

For example, both Borda and Plurality violate the IIA axiom in some cases, but the severity and structure of the IIA violations (and the kinds of profile on which they occur) can differ substantially. By shifting from pass/fail fairness conditions to continuous, interpretable metrics that capture the magnitude and profile-specific structure of axiom noncompliance, we can hope to select voting rules that are best adapted to particular uses.

Both metrics take a voting rule $f$ and profile $P$ as input and output values in the range $[0, 1]$, with higher scores indicating stronger adherence to the respective axioms.

\begin{itemize}
    \item \textbf{$ \SIIA(f, P)$}: A metric using Kendall tau (swap) distance between output rankings to capture the stability of $f$ when arbitrary candidates are removed from the election.
    \item \textbf{$ \SU(f, P)$}: A metric that quantifies the degree to which the voting outcomes respect majority preferences.  
\end{itemize}

\textbf{Contributions.} Our main contributions are as follows:

\begin{itemize}
    \item \textbf{Relaxed Axioms:} We define and motivate $\SIIA$ and $\SU$, then derive their basic properties.

    \item  \textbf{Quantitative Arrow's Theorem:} We establish a quantitative analog of Arrow's impossibility theorem: any rule satisfying $ \SIIA = 1$ and $ \SU > 0 $ for all profiles must be dictatorial. This establishes the axiomatic faithfulness of our metrics.

    \item \textbf{Empirical Evaluation:} We apply our framework to two datasets: (i) a real-world corpus of 1,070 ranked-choice elections from Scottish local governments, and (ii) synthetic profiles generated under a Bradley--Terry-based model. 
\end{itemize}

\subsection{Related Work} \label{subsec:related}

Quantifying axiom violations---rather than merely detecting their presence---has emerged as a central tool in multiple domains, including large language models   and algorithmic fairness. 
Recent works  have aimed to move beyond a pass/fail framework by developing real-valued metrics that capture the degree of deviation from normative principles. Similar efforts appear in machine learning, where  metrics have been proposed to quantify violations of group equity in classification, individual fairness in decision systems, and robustness in multi-agent planning.
This shift toward gradated rather than discrete evaluation has also found resonance in the context of large language models (LLMs), where researchers have begun to assess the extent to which model output violates classical axioms of rationality and fairness.
In AI alignment, and in reinforcement learning with human feedback (RLHF), preference aggregation plays a central role, yet the stability and consistency of these aggregations remain poorly understood. As LLMs and AI agents increasingly participate in deliberative and decision-making tasks, the need for principled, interpretable metrics to evaluate axiom adherence---especially under collective settings---has been flagged by researchers such as \citet{conitzer2024social,zhao2024electoral,zhao2024measuring}.

Within the more voting-focused setting of computational social choice theory, researchers such as  \citet{dougherty2020probability} analyze the frequency of (binary) IIA violations under simulated or real-world profiles, contrasting the stability of rules like Borda and Plurality.
Going beyond this, other new work has explored relaxed versions of IIA.  Notably, Eric Maskin  recently introduced a modification called MIIA \citep{maskin2025borda}.  Where IIA requires that whether $A\succ B$ in the output ranking depends only on the $A$ vs.~$B$ preference of individual voters, MIIA also allows for dependence on the number of candidates ranked {\em between} them, thought of as an indicator of preference strength.  However, satisfying MIIA remains a binary yes/no condition.   
Social choice researchers have also considered the frequency of rule violations for other axioms like monotonocity \citep{mccune2024monotonicity} within a binary framework.

\citet{delemazure2024independence} and \citet{zhao2024measuring} may provide the closest fit for the spirit of our contribution in this paper, because they attempt to quantify violation magnitude as well as frequency.  Delemazure and colleagues introduce a notion of {\em distortion} to measure how the removal of alternatives impacts social welfare. Still, this differs in two key ways from our work.  First, it requires either embeddings in a latent metric space or a measure of utility, so it cannot be computed directly from ranked preference profiles (which is the classical setting for Arrow's Theorem).  Second, and relatedly, the reliance on latent data requires the use of simulated elections rather than observed elections, and the simulations are conducted with uniform distributions that are widely recognized to be highly unrealistic.\footnote{Those uniform models are called IC and IAC.  In terms of their realism, see for instance \citet{tideman2010structure}, which explicitly includes the assessment ``None of the 11 models discussed so far [including IC, IAC] are based on the belief that the associated distributions of [preference profiles] might actually describe rankings in actual elections."}

Our IIA metric has a strong family resemblance to the metric introduced by \citet{zhao2024measuring}. 
Our $\SIIA$ measures stability/instability by holding out candidates one at a time.  We measure the swap distance between two alternatives:  running the voting rule on the full profile, then removing a candidate; or  removing the candidate first, then running the voting rule. Their similarity score uses edit distance in place of swap distance, but is otherwise similar.  
We argue that swap distance is more appropriate for rankings:  edit distance does not distinguish between swapping your third and fourth choices or your first and tenth.
Further, that paper restricted its findings to one table of statistics that focused on the removal of ``gold" (ground-truth) answers from rankings.  The analysis here is considerably more extensive.

The remainder of the paper is organized as follows. Section~\ref{sec:background} reviews the social choice background and introduces the needed notation. Section~\ref{sec:QFM} defines the relaxed axioms, analyzes their properties, and verifies the connection to Arrow's Theorem. Section~\ref{sec:empirical} presents empirical results on both observed and synthetic data. Finally, Section~\ref{sec:discussion} concludes with discussion, limitations, and  directions for future work.

\section{Background}\label{sec:background}

\subsection{Notation for elections, preferences, and preference profiles}
First, we give some notation for the elections we will consider.  We let $\voters$ be the set of $n$ voters and  $\candidates$  the set of $m$ candidates.  Each voter will offer a complete or partial ranking of the candidates as their 
{\em ballot}.  We will write $S(\candidates)$ for the set of permutations of the candidates, which is a copy of the symmetric group $S_m$, and we will write $\hat{S}(\candidates)$ for the extended set of partial rankings.  For instance, if $\candidates=\{A,B,C\}$, then $(A,B,C)$ is an example of a complete ranking and $(A)$ is a partial ranking, both valid ballots.\footnote{We do not allow ties, except implicitly in the sense that partial ballots may be read as giving equal rankings to the non-mentioned candidates.} 
We will sometimes use the alternative notation 
$A\succ_i B$ to mean that voter $i$ ranks $A$ above $B$, and $A\asymp_i B$ to mean that the voter leaves both unranked. Then $A\succeq_i B$ means $A\succ_i B$ or $A\asymp_i B$.  In order to measure the difference between rankings, a natural choice of metric is the Kendall tau or swap distance, which we will denote by $\dswap$.  This measures the number of neighbor swaps necessary to transform one ranking to another, and extends naturally to partial rankings.

The set of  {\em preference profiles} (or simply {\em profiles}) for elections on $\voters,\candidates$ is given by 
$\profiles:=\hat{S}(\candidates)^\voters$, which is a record of a ballot for each voter.
With this notation, a {\em voting rule} is a function $f:\profiles\to S(\candidates)$.  That is, even though the ballots are allowed to be partial rankings, the voting rule must output a complete ranking in a deterministic way, such as by using a tie-breaking procedure.
It will sometimes be convenient to use the notation $A\succ_{f(P)}B$ to mean that when the voting rule is applied to the profile $P$, the outcome ranks $A$ over $B$.

\begin{definition}[Relative ranking vector]
    Given a profile $P\in\profiles$ and candidates
    $A,B\in\C$, we let $X_{A,B}(P)\in\RR^n$ be the {\em relative ranking vector} given by
    \[
        X_{A,B}(P)_i:=\begin{cases}
            1 & A\succ_i B\\
            1/2 & A\asymp_i B\\
            0 & B\succ_i A.
        \end{cases}
    \]
\end{definition}

\begin{lemma}[Properties of relative ranking vector] 
For any profile $P\in \profiles$, \
    $(X_{A,B}+X_{B,A})\cdot \mathbbm{1}=n$. 
 On the other hand, $(X_{A,B}-X_{B,A})\cdot \mathbbm{1}$ is the (signed) margin with which the voters overall prefer $A$ to $B$.
\end{lemma}

\begin{proof}
For the sum, note that a voter $i$ that ranks at least one of $A,B$ contributes $1$ in the $i$th coordinate to exactly one of $X_{A,B}$ and $X_{B,A}$, while a voter that leaves both unmentioned contributes $1/2$ to each term.  Thus the sum is the all-ones vector.

The difference has entries in $\{1,-1,0\}$ according to whether the $i$th voter prefers $A$ to $B$, prefers $B$ to $A$, or ranked neither. Thus, summing its entries gives the signed margin.
\end{proof}

\begin{definition}[Alignment] The {\em alignment} between a voting rule $f$ and a profile $P$ on $A,B\in\C$ is measured by the share of voters ranking $A,B$ the same way as the output ranking $f(P)$, with a tie counting as half agreement:
\[
    I_{A,B}(f,P):=\begin{cases}
    \frac 1n\|X_{A,B}(P)\|_1, & A\succ_{f(P)} B\\
    \frac 1n\|X_{B,A}(P)\|_1, & B\succ_{f(P)} A.
    \end{cases}
\]
Then we can write an overall misalignment score  as 
$M(f,P):=\min\limits_{A,B\in\C} I_{A,B}(f,P)$.  
\end{definition}

This $M$ finds the worst alignment of the voters with the outcome.  In the case of complete rankings, this flags an anti-majoritarian outcome if $M<\frac 12$.

\begin{definition}[Disqualification]
Given a profile $P\in\profiles$ or a ballot $\beta\in \hat{S}(\C)$ and a candidate $C\in\C$, we write $P^C$ or $\beta^C$ to represent the condensed profile or ballot on $\C'=\C\setminus\{C\}$, i.e., the candidate set without $C$. 
\end{definition}

Finally, we record the margins of victory in a complete graph on the candidates and we employ a definition from graph theory that picks out rankings of candidates that respect majority preferences.

\begin{definition}[Pairwise comparison graph and topological sort]\label{def:topsort}
For any profile $P$, the pairwise comparison graph (or PWCG)  is a directed weighted graph $G(P)$ with vertex set $V=\C$ defined by putting a directed edge from $A$ to $B$ with weight $(X_{A,B}-X_{B,A})\cdot \mathbbm{1}$, which is the margin by which voters prefer $A$ to $B$.  

We say that a complete ranking of vertices, 
$\pi\in S(V)$, is a {\em topological sort} of a directed graph $G=(V,E)$ if $(u,v)\in E \implies u\succ v$ in $\pi$.  
\end{definition}

\subsection{Voting rules}
In this paper, we focus on a small menu of common ranking-based voting rules. The first several rules are in the class called {\em scoring rules}:  a score vector $(s_1,s_2,\dots,s_m)$ is applied, so that first-place votes are worth $s_1$ points, second-place votes are worth $s_2$, and so on. For simplicity, we adopt the convention that candidates who are unranked in a given ballot receive no points from that ballot. Then the output ranking given by $f$ is just a list of the candidates in order of their total score from all voters. Table \ref{tab:scoring-rules} summarizes the scoring rules compared below.

\begin{table}[htb!]
    \centering
    \begin{tabular}{ccl}
Voting rule & Score vector & Description\\
\hline 
Borda & $(m,m-1,\dots,1)$ & Linear scoring. \\
$3$-approval & $(1,1,1,0,\dots,0)$ & Top three get a point.\\
$2$-approval & $(1,1,0,0,\dots,0)$ & Top two get a point.\\
Plurality  & $(1,0,0,0,\dots,0)$ & Only first choice gets a point.\\       
    \end{tabular}
    \caption{Scoring rules.}
    \label{tab:scoring-rules}
\end{table}

Finally, we compare these with another voting rule that cannot be described with a score vector, but rather operates iteratively in a round-by-round fashion.
Single transferable vote (or \textbf{STV}) is a multi-winner voting rule for a fixed number $k$ of seats, designed to output a set of $k$ winners.
A {\em threshold of election} is also fixed---often, this is the so-called {\em Droop quota} $n/(k+1)$---as the amount of support needed to be elected.  
If any candidate has strictly more than that level of first-place support in the profile $P$, they are elected.  If no candidate does, then proceeding round by round, we eliminate the candidate with the least first-place support and {\em transfer} their support to the next choice of their voters. When a candidate is elected, their surplus support is transferred to the next choice of their voters; for instance, if they received 120\% of the threshold, then their votes are transferred with weight $20/120=1/6$.  
Though STV is designed as a multi-winner system, we report the results as a ranking by filling in winners from the top of the ranking in order of election, and filling in eliminated candidates from the bottom in order of elimination.  Any candidates still left once the seats have been filled are placed in between, in order of first-place votes when the process terminates.\footnote{For further details,  see the documentation for the {\sf VoteKit} Python package, which was used in the empirical sections of this paper \citep{VoteKit}.}
 
\subsection{Observed and synthetic elections}

We apply our framework to a real-world dataset of Scottish ranked-choice elections, comprising over 1000 local elections.\footnote{This dataset was  collected by David McCune and published by the Data and Democracy Lab.  It can be found at  \url{https://github.com/mggg/scot-elex}.} For administrative purposes, Scotland is divided into several hundred wards, each conducting a ranked choice election by STV for its local government every five years since 2012.
An example is shown in Figure~\ref{fig:pwcg_ren_2022_ward2} with the candidates listed  as they appeared on the ballots (alphabetically by last name) and the PWCG displayed at right.

\begin{figure}[htb!]\centering 
\begin{tikzpicture}
\node at (0,1) [right] {1: Edward Grady (Labour)};
\node at (0,.5) [right] {2: Kate Hughes (Labour)};
\node at (0,0) [right] {3: Cathy McEwan (SNP)};
\node at (0,-.5) [right] {4: Dale Nelson (Conservative)};
\node at (0,-1) [right] {5: Jim Paterson (SNP)};
\node at (8,0) {\includegraphics[width=2in]{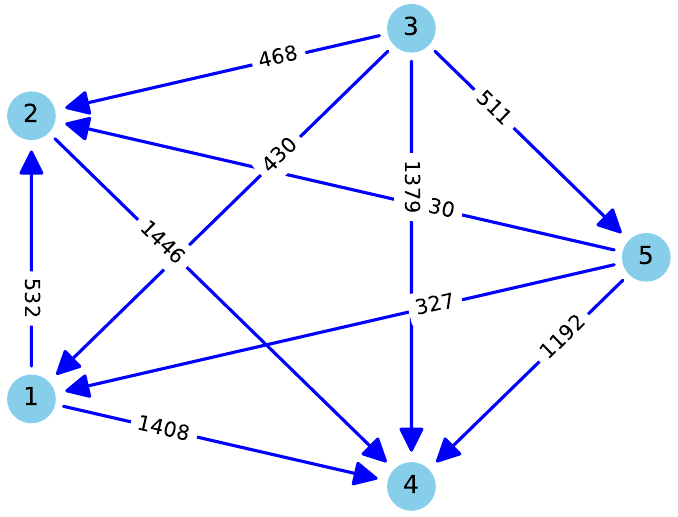}};
\end{tikzpicture}    
    \caption{The pairwise comparison graph for the Scottish local government election in Renfrewshire Ward 2, 2022. In this case, Cathy McEwan is the Condorcet candidate, preferred head-to-head to all others. Below, Table~\ref{tab:Renf-results} will show the results when various voting rules are applied.}
    \label{fig:pwcg_ren_2022_ward2}
\end{figure}

\section{Relaxed axioms and their properties}\label{sec:QFM}

We propose two real-valued metrics that can be evaluated on particular preference profiles, relaxing classical axioms.

\subsection{Independence of Irrelevant Alternatives (IIA)}

We first extend the axiom of Independence of Irrelevant Alternatives (IIA). Informally, a voting rule satisfies IIA if the relative position of any two candidates in the social ranking only depends on their relative ranking by individual voters.

\begin{definition}[IIA and $\SIIA$]

A voting rule $f$ satisfies IIA if for any two profiles $P,P'\in \profiles$ and any $A,B\in\candidates$,
    \[
    X_{A,B}(P)=X_{A,B}(P')\implies
         \left(A\succ_{f(P)} B \iff A\succ_{f(P')}B\right).
    \]
Equivalently, $f$ satisfies IIA iff $f(P^C)=f(P)^C$ for all $P\in\profiles, C\in\candidates$.

We define $\SIIA(f,P)$ to measure the failure of IIA:
\[
    \SIIA(f,P) := 1 - \frac{\sum\limits_{C \in\candidates} \dswap (f(P^{C}), f(P)^C)}{m \binom{m-1}{2}}.
\]

\end{definition}

Note that the largest possible swap distance between vectors of length $m$ comes from total reversal, which gives $\dswap=\binom{m-1}2$.  
From this characterization, we obtain some immediate consequences.
\vspace{.5cm}
\begin{prop}[IIA vs. $\SIIA$] \label{thm:IIA}
The IIA property is equivalent to $\SIIA(f,P)\equiv 1$ over $P\in\profiles$.

Furthermore, $\SIIA(f, P) = 0$ iff $f(P^C)$ is the complete reversal of $f(P)^C$ for every candidate $C\in\C$.
\end{prop}

Going further, the meaning of $\SIIA(f,P)=1-\alpha$ is that, on average, disqualifying a candidate alters the ranking of the remaining candidates by a fraction $\alpha$ of a complete reversal.

\subsection{Unanimity  (U)}

The other axiom we extend is the unanimity criterion (U). Informally, a voting rule $f$ satisfies U if in the case that $P$ contains a unanimous preference for some candidate over another one, the output $f(P)$ respects that preference.
Note that this says nothing about profiles in which there are no unanimous preferences---which means this axiom is silent in most real-world elections. This motivates us to create a metric that also gives nontrivial information in cases where the majority preference is not unanimous.

\begin{definition}[U and $\SU$]
    A voting rule $f$ satisfies U if for any  $P\in\profiles$ and $A,B\in \C$, if $A\succ_i B$ for all voters $i\in\voters$, then $A\succ_{f(P)} B$.

    Equivalently, 
    \[f ~\hbox{\rm satisfies~U} \iff 
    I_{A,B}(f,P)>0 \quad \forall P,A,B
    \iff 
    M(f,P)>0 \quad \forall P.\]

Then we can easily define a metric of majoritarianism that extends the unanimity criterion:
\[
\SU(f,P):=\begin{cases}
        \frac{M}{1-M}, & M(f,P)<1/2\\
        1, & M(f,P)\ge 1/2.
    \end{cases}
\]
\end{definition}

This is set up so that $\SU(f,P)=0$ when $P$ is a witness to the failure of the unanimity criterion.  In other words, $\SU(f,P)=0$ exactly means that $P$ has a unanimous preference between some two candidates and $f(P)$ ranks them in reverse.

By contrast, the definition makes it easy to see that $\sigma_U(f,P)= 1$ precisely means that every majority preference in $P$ is respected.  (And clearly this implies that every unanimous preference is respected.)

This description also gives a natural interpretation to scores less than one.  If $\SU=\frac{\alpha}{1-\alpha}<1$, then $\alpha<1/2$, and there is some pair of candidates where only $\alpha$ share of voters agrees with the output ranking.

\subsection{Optimization of U}

One crucial observation is that $\SU(f,P)$ can be computed from only $P$ and $f(P)$.   This means that it is feasible to efficiently construct a voting rule $f$ that is engineered for optimal performance on $\SU$.
By contrast, computing $\SIIA(f,P)$ requires knowledge of not only $f(P)$, but also 
$f(P^C)$ for each $C\in\C$.

The idea of optimizing U on a given profile is simple: if the graph has no directed cycles (called {\em Condorcet cycles} in social choice theory), then any topological sort (Definition~\ref{def:topsort}) is a valid $f(P)$ that achieves a perfect $\SU=1$.  
However, if there are cycles, then no complete ranking can be fully consistent with the directed edges.  
In that case, a best-possible ranking is found by successively deleting all edges of minimal weight until a topological sort exists.  
Topological sorts can be executed in polynomial time with Kahn's algorithm.  This allows us to state the following result, with the precise algorithm formulation and proof of runtime deferred to the supplemental material.
\vspace{.5cm}
\begin{thm}[Greedy sort]\label{thm:uf_alg}
The greedy algorithm described above computes a voting rule $f_U$ that maximizes $\SU$ on all profiles, in the sense that $\sigma(f_U,P)\ge \sigma(g,P)$ for all voting rules $g$ and profiles $P$.
The algorithm to compute $f_U(P)$ runs in time $O(m^2n+m^3)$.
\end{thm}

\subsection{Quantitative Arrow's Theorem}

Arrow’s classical impossibility theorem \citep{arrow1950difficulty} asserts that no voting rule can satisfy IIA, U, and non-dictatorship simultaneously when there are at least three candidates. Traditionally, these axioms are treated in a binary fashion, either fully satisfied or not.  Our metrics $\SIIA$ and $\SU$ allow us to write a quantitative version of Arrow's Theorem. This reformulation maintains the Arrow's Theorem's essential message while enabling empirical and comparative analysis of voting rules. We state and prove the result below.

\textbf{Quantitative Arrow's Theorem:} A voting rule $f$ is a dictatorship if and only if $\SIIA(f,P) = 1$ and $\SU(f,P) > 0$ for all profiles $P$.

\begin{proof}
Suppose that $ \SIIA(f, P) = 1$ for all profiles $P$. Then, by Proposition 2, $f$ satisfies Arrow's IIA axiom in the classical (binary) sense. Furthermore, if $ \SU(f, P) > 0$ for all $P$, then $f$ does not violate Unanimity. Together, these conditions imply that $f$ satisfies both IIA and U in Arrow's original framework. By Arrow's impossibility theorem it follows that $f$ must be dictatorial.

For the converse, assume that $f$ is a dictatorship. Then, by definition, $f$ satisfies IIA and U in the binary (axiomatic) sense. By Proposition 2, this implies that $ \SIIA(f, P) = 1$ and since $f$ do not violate Unanimity, we get $ \SU(f, P) > 0$ for all profiles $P$.
\end{proof}

\section{Empirical Results}\label{sec:empirical}

\subsection{Scottish elections}

In Table~\ref{tab:Renf-results}, we return to the 2022 Scottish election of the Renfrewshire Ward 2 first shown in Figure~\ref{fig:pwcg_ren_2022_ward2}.

\begin{table}[htb!]
    \centering
    {\small \begin{tabular}{c|c|c|c|c|c}
         & \textbf{Borda} & \textbf{3-App} & \textbf{2-App} & \textbf{Plurality} & \textbf{STV} \\
        \hline \hline 
        1st & 3 & 3 & 3 & 3 & 3 \\
        2nd & 5 & 2 & 5 & 1 & 1 \\
        3rd & 1 & 5 & 1 & 5 & 5 \\ 
        4th & 2 & 1 & 2 & 4 & 4 \\
        5th & 4 & 4 & 4 & 2 & 2 \\
        \hline \hline 
        $\SIIA$ & 0.93 & 0.90 & 0.93 & 0.80 & 0.73 \\
        $\SU$   & 1.00 & 0.75 & 1.00 & 0.44 & 0.44 
    \end{tabular}}
    \caption{Using the preference profile from Renfrewshire Ward 2, 2022, we compute the output rankings and the $\sigma$ metrics under the five different systems of election considered in this paper. Borda and 2-approval give the same outcome, which respects all majority preferences in this profile.  Note that Plurality and STV give the same output ranking, which ensures that the $\SU$ values are equal, but they perform differently when candidates are removed, leading to different $\SIIA$.}
    \label{tab:Renf-results}
\end{table}

Next, we sweep our scores over the full Scottish dataset and show the results in Figures~\ref{fig:scottish-boxplots}--\ref{fig:scottish-maps}.
We find that the Borda rule most frequently receives the highest  $\SIIA$ and $\SU$ scores among the alternatives.

\begin{figure}[htb!]
  \centering
  \includegraphics[width=5in]{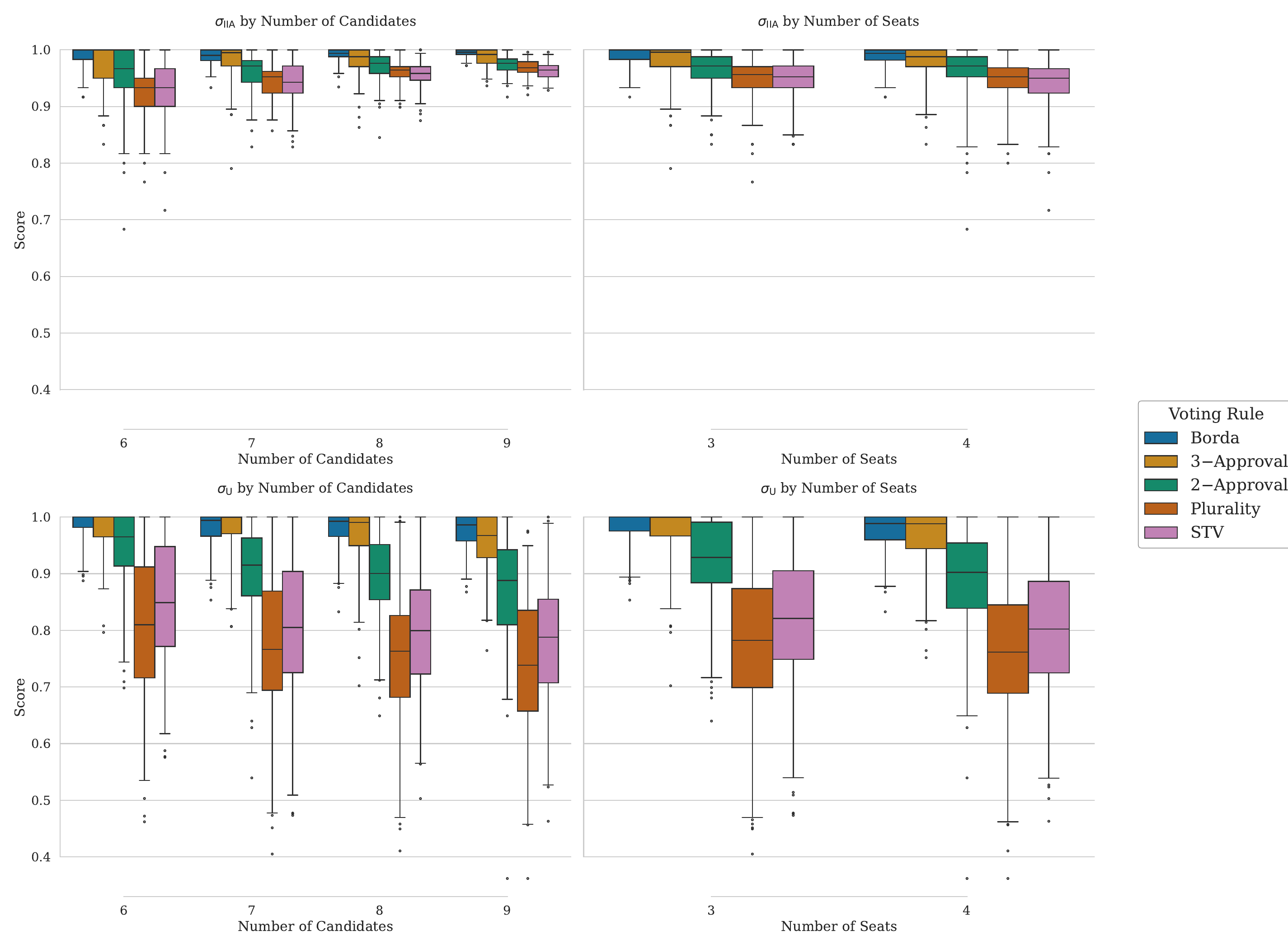}
  \caption{
Boxplots of the $\SIIA$ and $\SU$ metrics for the five voting rules on the Scottish elections. Most elections have 6-9 candidates and are used to select 3-4 winners, and the plots are split out accordingly. As we might expect, elections with more candidates are more stable, but tend to contain more violations of pairwise preferences.}
\label{fig:scottish-boxplots}
\end{figure}

\begin{figure}[htb!]
  \centering
  \includegraphics[width=5in]{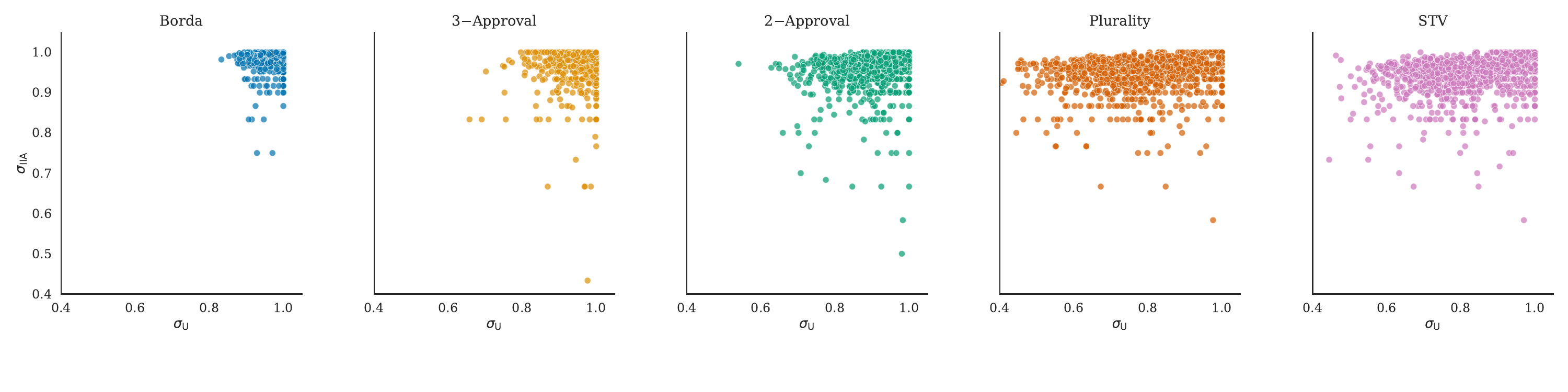}
  \caption{
Scatterplots of  $\SIIA$ and $\SU$  on the Scottish data help us visualize the relationship between the two metrics.
}
\label{fig:scottish-scatter}
\end{figure}

\clearpage

\begin{figure}[htb!]
  \centering
  \includegraphics[width=6in]{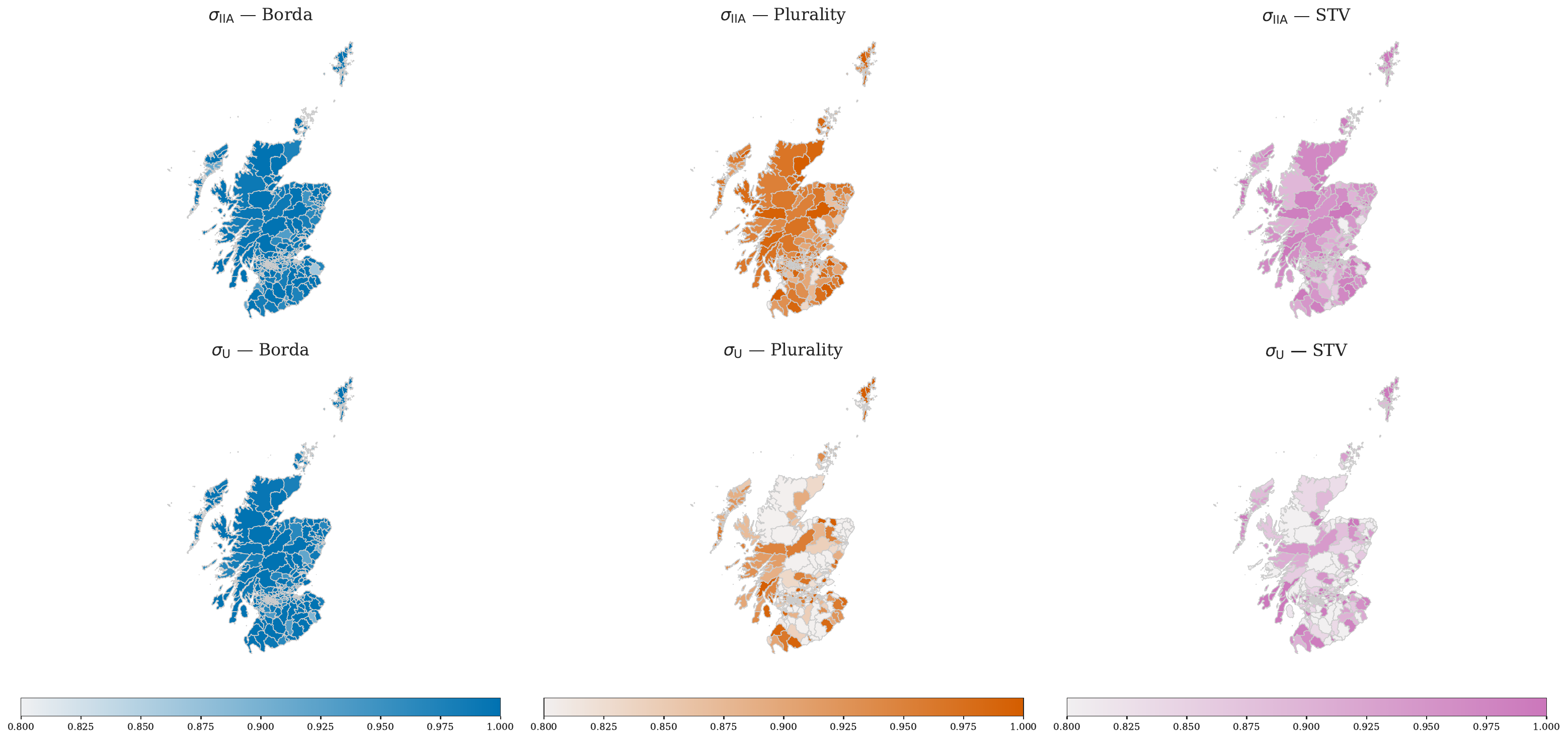}
  \caption{
Ward-level heatmaps showing the $\SIIA$ and $\SU$ metrics for Borda, Plurality, and STV, applied to the 2012 Scottish local elections. 
}
\label{fig:scottish-maps}
\end{figure}

\subsection{Synthetic data}
This section reports the outcomes of trials run on  synthetic preference profiles generated using a model derived from a Bradley–Terry (BT) process.  Bradley-Terry is a classical statistical model for generating ranked preferences based on latent candidate strengths. These trials include 100 independent profiles of 1000 ballots each for every pictured box.

One of the parameters in this model is a Dirichlet parameter $\alpha$ representing candidate strength, where $\alpha\to\infty$ tends to produce equal latent utilities for the candidates, and $\alpha\to 0$ gives a single candidate utility approaching 100\% for the voters.  We call these {\em uniform} and {\em dominant} scenarios for candidate strength.  Figure~\ref{fig:BT-boxplots} shows the results.

\begin{figure}[htb!]
  \centering
\includegraphics[width=5in]{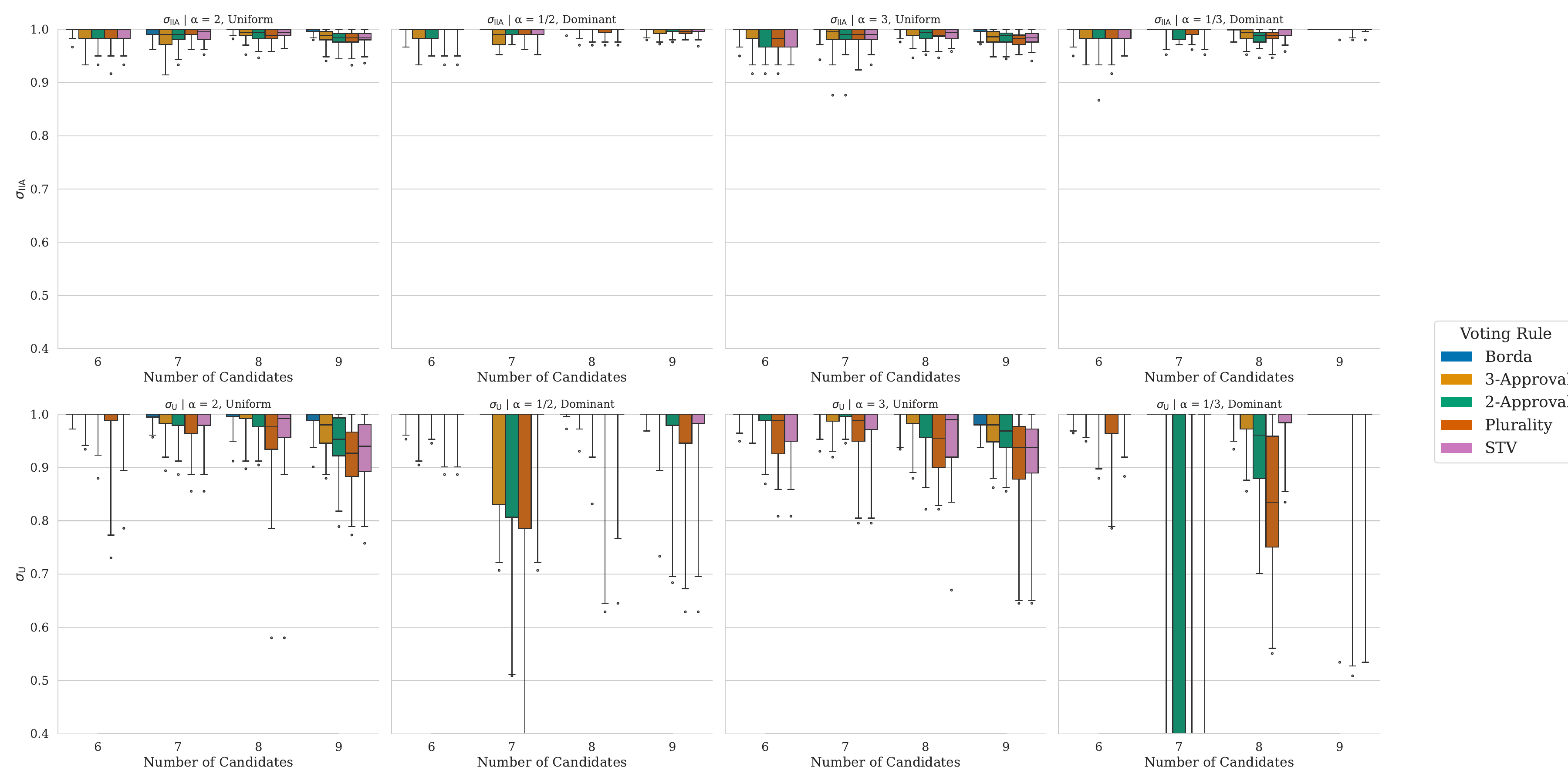}
\caption{Boxplots of $\SIIA$ and $\SU$ for the five voting rules applied to profiles generated from a Bradley-Terry model. We use Dirichlet parameters to vary candidate strength:  $\alpha=2,3$ give  scenarios in which voters tend to have more uniform preferences among candidates, while $\alpha=1/2,1/3$ moves towards  scenarios with a strong candidate preference. We have kept the $[0.4,1]$ range as in other plots for ease of comparison, though a few $\SU$ outputs fall outside the window. }\label{fig:BT-boxplots}
\end{figure}

\section{Discussion and Significance}\label{sec:discussion}
In this paper, we introduce quantitative metrics that relax two foundational social choice axioms: Independence of Irrelevant Alternatives (IIA) and the Unanimity criterion (U).
Our analysis demonstrates the value of extending classical binary axioms to gradated metrics, enabling the assessment of voting rules along a continuum of axiom compliance.
Across both real-world elections from Scottish local government and synthetic  datasets generated with a Bradley-Terry statistical model, the Borda rule consistently achieves strong marks in its stability and its tendency to comport with majority preferences.  
The framework also recovers Arrow's Impossibility Theorem in a metric setting, providing a new lens for interpreting the trade-offs between stability and majoritarianism.  

However, readers will note that we have refrained from using the common phrasing in which IIA and U are called axioms of {\em fairness}.  This analysis has led us to a more nuanced view of the desirability of the axioms under discussion.  While multi-agent cooperation may well benefit from stability and a majoritarian tilt, it is less obvious in the setting of political representation, where these attributes could easily cut against other important normative values, like {\em responsiveness} (the tendency for outcomes to shift to reflect changing preferences) and {\em proportionality} (which in particular requires that sizable minorities receive nontrivial representation).

We believe that this work may be of independent interest for quantifying axiom compliance in welfare economics and in applications involving large language models (LLMs), where incorporating more axioms in the regularization objectives---rather than focusing solely on utility or main accuracy loss---can lead to more balanced and principled decision-making.

In particular, \cite{mishra2023ai} examine the challenges of utilizing democratic processes, such as reinforcement learning from human feedback, to align AI systems with human values. They argue that no voting rule can fairly and consistently aggregate diverse preferences without violating essential principles such as fairness and individual rights.  In contrast, we propose new metrics that relax Arrow’s axioms, allowing researchers to evaluate how closely a voting rule adheres to principles such as IIA and U. Instead of treating these principles as binary pass/fail conditions, we suggest continuous measures that can be applied in real-world contexts, including AI systems. The article by \cite{mishra2023ai} illustrates the limitations of what can be achieved perfectly, whereas we provide tools to navigate these limitations as effectively as possible.

\subsection{Limitations and Future Work}

The definitions offered here are currently limited to voting rules that produce complete, tie-free rankings over all candidates. 
Importantly, the STV elections that generated the Scottish data were actually run to produce winner sets of fixed size (typically $k=3,4$, as mentioned above). This may be reflected in voting behavior; for instance, Scottish ballots in $k$-winner contests often rank only $k$ candidates and leave the other positions blank.  Therefore, representing STV output as a ranking, while meaningful, is not fully aligned with the nature of the elections where it was used.

Ideas for future work include the following.
\begin{itemize}
    \item Extend definitions to accommodate generalized ballots (with ties in arbitrary positions) as input and output.  This allows for a more natural representation of STV output as a winner {\em set} rather than ranking.
    \item Consider other binary axioms that are relevant in voting theory and machine learning, like monotonicity.
    \item Extend experiments to more fully explore the Pareto frontier for $(\SIIA,\SU)$.
\end{itemize}

\section{Acknowledgment}

The authors thank Erica Chiang and Sujai Hiremath for helpful early discussions, and Peter Rock for assistance with coding-related technical concepts.

\bibliographystyle{plainnat}
\bibliography{bibliography}

\begin{thebibliography}{12}
\providecommand{\natexlab}[1]{#1}
\providecommand{\url}[1]{\texttt{#1}}
\expandafter\ifx\csname urlstyle\endcsname\relax
  \providecommand{\doi}[1]{doi: #1}\else
  \providecommand{\doi}{doi: \begingroup \urlstyle{rm}\Url}\fi

\bibitem[Arrow(1950)]{arrow1950difficulty}
Kenneth~J Arrow.
\newblock A difficulty in the concept of social welfare.
\newblock \emph{Journal of Political Economy}, 58\penalty0 (4):\penalty0
  328--346, 1950.

\bibitem[Conitzer et~al.(2024)Conitzer, Freedman, Heitzig, Holliday, Jacobs,
  Lambert, Moss{\'e}, Pacuit, Russell, Schoelkopf, et~al.]{conitzer2024social}
Vincent Conitzer, Rachel Freedman, Jobst Heitzig, Wesley~H Holliday, Bob~M
  Jacobs, Nathan Lambert, Milan Moss{\'e}, Eric Pacuit, Stuart Russell, Hailey
  Schoelkopf, et~al.
\newblock Social choice should guide {AI} alignment in dealing with diverse
  human feedback.
\newblock \emph{arXiv preprint arXiv:2404.10271}, 2024.

\bibitem[{Data and Democracy Lab}(2024)]{VoteKit}
{Data and Democracy Lab}.
\newblock {VoteKit}: Python package.
\newblock \url{https://github.com/mggg/VoteKit}, 2024.
\newblock GitHub repository.

\bibitem[Delemazure et~al.(2024)Delemazure, Lang, and
  Pierczy{\'n}ski]{delemazure2024independence}
Th{\'e}o Delemazure, J{\'e}r{\^o}me Lang, and Grzegorz Pierczy{\'n}ski.
\newblock Independence of irrelevant alternatives under the lens of pairwise
  distortion.
\newblock In \emph{Proceedings of the AAAI Conference on Artificial
  Intelligence}, volume~38, pages 9645--9652, 2024.

\bibitem[Dougherty and Heckelman(2020)]{dougherty2020probability}
Keith~L Dougherty and Jac~C Heckelman.
\newblock The probability of violating {A}rrow’s conditions.
\newblock \emph{European Journal of Political Economy}, 65:\penalty0 101936,
  2020.

\bibitem[Efron and Tibshirani(1994)]{efron1994introduction}
Bradley Efron and Robert~J Tibshirani.
\newblock \emph{An Introduction to the Bootstrap}.
\newblock Chapman and Hall/CRC, 1994.

\bibitem[Maskin(2025)]{maskin2025borda}
Eric Maskin.
\newblock Borda’s rule and {A}rrow’s independence condition.
\newblock \emph{Journal of Political Economy}, 133\penalty0 (2):\penalty0
  385--420, 2025.

\bibitem[McCune and Graham-Squire(2024)]{mccune2024monotonicity}
David McCune and Adam Graham-Squire.
\newblock Monotonicity anomalies in {S}cottish local government elections.
\newblock \emph{Social Choice and Welfare}, 63\penalty0 (1):\penalty0 69--101,
  2024.

\bibitem[Mishra(2023)]{mishra2023ai}
Abhilash Mishra.
\newblock {AI} alignment and social choice: Fundamental limitations and policy
  implications.
\newblock \emph{arXiv preprint arXiv:2310.16048}, 2023.

\bibitem[Tideman and Plassmann(2010)]{tideman2010structure}
T~Nicolaus Tideman and Florenz Plassmann.
\newblock The structure of the election-generating universe.
\newblock Manuscript, 2010.

\bibitem[Zhao et~al.(2024{\natexlab{a}})Zhao, Wang, and
  Peng]{zhao2024electoral}
Xiutian Zhao, Ke~Wang, and Wei Peng.
\newblock An electoral approach to diversify {LLM}-based multi-agent collective
  decision-making.
\newblock \emph{arXiv preprint arXiv:2410.15168}, 2024{\natexlab{a}}.

\bibitem[Zhao et~al.(2024{\natexlab{b}})Zhao, Wang, and
  Peng]{zhao2024measuring}
Xiutian Zhao, Ke~Wang, and Wei Peng.
\newblock Measuring the inconsistency of large language models in preferential
  ranking.
\newblock \emph{arXiv preprint arXiv:2410.08851}, 2024{\natexlab{b}}.

\end{thebibliography}


\newpage

\renewcommand{\thesection}{\Alph{section}}
\counterwithin{equation}{section} 

\appendix
\textbf{Appendix}
\section{Proofs}

\textbf{Proof of Proposition 2}
\begin{proof}
For any candidates $A,B,C\in\candidates$ and profile $P\in\profiles$, $X_{A,B}(P)=X_{A,B}(P^C)$ because removing $C$ does not affect the relative order of any other candidates. Therefore, if $f$ satisfies IIA, $f(P^C)$ and $f(P)$ rank $A$ and $B$ in the same relative order for any pair of candidates $A,B\in\candidates$ with $A,B\neq C$. This implies $f(P^C)=f(P)^C$, which means $\SIIA(f,P)=1$.

Conversely, if $\SIIA(f,P)=1$ for all $P\in\profiles$, then it must be true that $f(P)^C=f(P^C)$ for all profiles $P\in\profiles$ and candidates $C\in\candidates$. By induction, this means that $f(P)^Z=f(P^Z)$ for any subset of candidates $Z\subseteq\candidates$, where we take $P^Z$ and $f(P)^Z$ to mean $P$ and $f(P)$ but with all candidates in $Z$ removed, respectively.

Now suppose that for some profiles $P,P'\in\profiles$ and candidates $A,B\in\candidates$, $X_{A,B}(P)=X_{A,B}(P')$. If we let $Z=\candidates\setminus\{A,B\}$, then $P^Z=(P')^Z$ because the remaining candidates in both profiles are just $A$ and $B$, and we already know the voters agree on their relative order. Therefore, using this and the fact we proved in the last paragraph, we have that
\[
    f(P)^Z=f(P^Z)=f((P')^Z)=f(P')^Z,
\]
which means that $f(P)$ and $f(P')$ rank $A$ and $B$ in the same order. Thus, $f$ satisfies IIA.



\end{proof}

\begin{algorithm}
\caption{Computes a $ \SU$-optimal voting rule $f$}\label{alg:uf_opt}
\begin{algorithmic}
\Procedure{$f$}{$P$}
    \State $G \gets \PWCGraph(P)$
    \State $T \gets \TopSort(G)$
    \While{$T=\textrm{None}$}
        \State $m \gets \textrm{min}([e.weight\textrm{ for } e \textrm{ in }G.edges])$
        \For{$e \in G.edges$}
            \If{$e.weight = m$}
                \State $G.remove(e)$
            \EndIf
        \EndFor
        \State $T \gets \TopSort(G)$
    \EndWhile
    \State \Return $T$
\EndProcedure
\end{algorithmic}
\end{algorithm}

In order to prove Theorem \ref{thm:uf_alg}, it will help to have the following notation. 

For any voting rule $f$ and profile $P$, we will let $D(f,P)$ denote the set of pairs of candidates for which there is an edge in $G(P)$ that disagrees with $f(P)$. By this definition, $\{A,B\}\in D(f,P)$ if and only if $I_{A,B}(f,P)<1/2$.

We will also let $\delta_{A,B}(P)$ be the margin, i.e. the weight of an edge in $G(P)$ between $A$ and $B$, or zero if there is no edge.

First, we formally show $\SU(f,P)$ is inversely related to the maximum margin edge in $G(P)$ that disagrees with $f(P)$. We then use this to prove Theorem $\ref{thm:uf_alg}$.





\begin{lemma}\label{lem:uf_and_margins}
If $P$ is a profile with a Condorcet cycle, then $\SU(f,P)$ is a strictly decreasing function of the maximum weight of any edges in $G(P)$ that disagree with $f(P)$, and this decreasing function is independent of $f$. 
\end{lemma}
\begin{proof}
Expanding out the definition of $\SU(f,P)$, we have that
\[
    \SU(f,P)=\min_{\{A,B\}\subseteq\candidates}\begin{cases}
        \frac{I_{A,B}(f,P)}{1-I_{A,B}(f,P)} & I_{A,B}(f,P)<1/2\\
        1 & o.w.
    \end{cases}.
\]
By the definition of $D(f,P)$, $\{A,B\}\in D(f,P)$ if and only if an edge between $A$ and $B$ in $G(P)$ disagrees with $f(P)$, which is true if and only if $I_{A,B}<1/2$. Since $P$ has a Condorcet cycle, $G(P)$ has a directed cycle, and so $D(f,P)$ cannot be empty, since this would imply that $f$ is a topological sort of $G(P)$. Therefore, since $\frac{I_{A,B}(f,P)}{1-I_{A,B}(f,P)}<1$ whenever $I_{A,B}(f,P)<1/2$ and $D(f,P)$ is non-empty, we must have that 
\[
    \min_{\{A,B\}\subseteq\candidates}\begin{cases}
        \frac{I_{A,B}(f,P)}{1-I_{A,B}(f,P)} & I_{A,B}(f,P)<1/2\\
        1 & o.w.
    \end{cases}=\min_{\{A,B\}\in D(f,P)}\frac{I_{A,B}(f,P)}{1-I_{A,B}(f,P)}.
\]
Now suppose that $\{A,B\}\in D(f,P)$. Then there is an edge in $G(P)$ that disagrees with $f(P)$, and so without loss of generality, $(A,B)$ is that edge and $B\succ_{f(P)}A$. This means that $I_{A,B}(f,P)=(1/n)\|X_{B,A}(P)\|_1$, and so the weight of edge $(A,B)$ is
\[
    \delta_{A,B}(P)=\|X_{A,B}(f,P)\|_1-\|X_{B,A}(f,P)\|_1=n-2\|X_{B,A}(f,P)\|=n-2nI_{A,B}(f,P).
\]
Thus, if we define $h(x):=1/2-x/2n$, then $I_{A,B}(f,P)=h(\delta_{A,B}(P))$. If we define $g(x):=x/(1-x)$, then we can see that
\[
    \sigma_{U}(f,P)=\min_{\{A,B\}\in D(f,P)}g(h(\delta_{A,B}(P))).
\]
Since $g$ is increasing on the interval $[0,1)$, $h$ is decreasing on the interval $(-\infty,\infty)$, and $h([0,n])=[0,1/2]\subset[0,1)$, $g\circ h$ is increasing on the interval $[0,n]$, and so since $\delta_{A,B}(P)\in[0,n]$, we must have that
\[
    \min_{\{A,B\}\in D(f,P)}g(h(\delta_{A,B}(P)))=g\left(h\left(\max_{\{A,B\}\in D(f,P)}\delta_{A,B}(P)\right)\right).
\]
\end{proof}

Algorithm \ref{alg:uf_opt} formalizes the algorithm described in the main paper. Here, $\PWCGraph(P)$ is a procedure that computes the pairwise comparison graph given a profile $P$, and this runs in time $O(m^2n)$ since it just computes $\binom{m}{2}$ margins, each of which is a sum of $n$ voters' preferences. Also, $\TopSort(G)$ is a procedure that computes a topological sort of a graph $G$ if one exists, and otherwise returns ``None'', and the example algorithm that we choose in this case is Kahn's algorithm, which runs in time $O(|V(G)|+|E(G)|)$, where $V(G)$ and $E(G)$ are the vertex and edge sets of $G$, respectively.

\vspace{.5cm}
\textbf{Proof of Theorem \ref{thm:uf_alg}:}

\begin{proof}


\textbf{Part 1}: Algorithm \ref{alg:uf_opt} runs in time $O(m^2n+m^3)$.\\
This algorithm runs in time $O(m^2n+m^3)$ since $\PWCGraph(P)$ runs in time $O(m^2n)$, $\TopSort(G)$ runs in time $O(m^2)$ (since the number of vertices in $G(P)$ is $m$ and the number of edges is at most $\binom{m}{2}$), and there can be at most $m$ iterations of the while loop since at least one edge is removed from $G$ each iteration.

This algorithm also always outputs a ranking of the candidates because the while loop must end, since if we remove all the edges of $G$, then any ordering of the candidates is a topological sort of the remaining graph.\\\newline
\textbf{Part 2}: The voting rule that Algorithm \ref{alg:uf_opt} computes maximizes $\SU$ on all profiles.\\
Let $f_U$ be the voting rule that Algorithm \ref{alg:uf_opt} computes. Suppose for the sake of contradiction that there exists some profile $P$ and some other voting rule $g$ for which $g$ has a strictly higher $\SU$ score, or in other words,
\[
    \SU(f_U,P)<\SU(g,P).
\]
With loss of generality, we can also suppose that $g$ maximizes $\SU$ on profile $P$. Since $\SU(\cdot,\cdot)\in[0,1]$, this means that $\SU(f_U,P)<1$. This means that $P$ must have some Condorcet cycle, since if it did not, then $G(P)$ would have a topological sort, and so $f_U(P)$ would be one of those topological sorts and we would have that $\SU(f_U,P)=1$. 


Therefore, by Lemma \ref{lem:uf_and_margins}, $\SU(f_U,P)$ is a strictly decreasing function of $\max_{\{A,B\}\in D(f,P)}\delta_{A,B}(P)$ that does not depend on $f_U$, and the same is true for $g$. Since $\SU(f_U,P)<\SU(g,P)$, this means that
\begin{equation}\label{eqn:margin_inequality}
    \max_{\{A,B\}\in D(f_U,P)}\delta_{A,B}(P)>\max_{\{A,B\}\in D(g,P)}\delta_{A,B}(P).
\end{equation}

Algorithm \ref{alg:uf_opt} cannot exit the while loop before all edges that disagree with $f_U(P)$ have been removed from $G(P)$, since otherwise $f_U(P)$ cannot be a topological sort of the remaining graph at that point.

Furthermore, since Algorithm \ref{alg:uf_opt} removes edges in increasing weight order and in rounds of all edges of a particular weight and equation \ref{eqn:margin_inequality} implies that there is at least one edge in $G(P)$ that disagrees with $f_U(P)$ and whose weight is strictly larger than that of all edges that disagree with $g(P)$, all edges that disagree with $g(P)$ will be removed before the while loop exits, and there will be at least one more iteration of the while loop right after all edges that disagree with $g(P)$ have been removed.

However, right after all edges that disagree with $g(P)$ have been removed, $g(P)$ must be a topological sort of the remaining graph. This is a contradiction, since it means that the algorithm would exit the while loop at this point, and we just reasoned above that the while loop must continue after this point.

\end{proof}

\begin{figure}[H]

  \centering
  \includegraphics[width=0.99\linewidth]{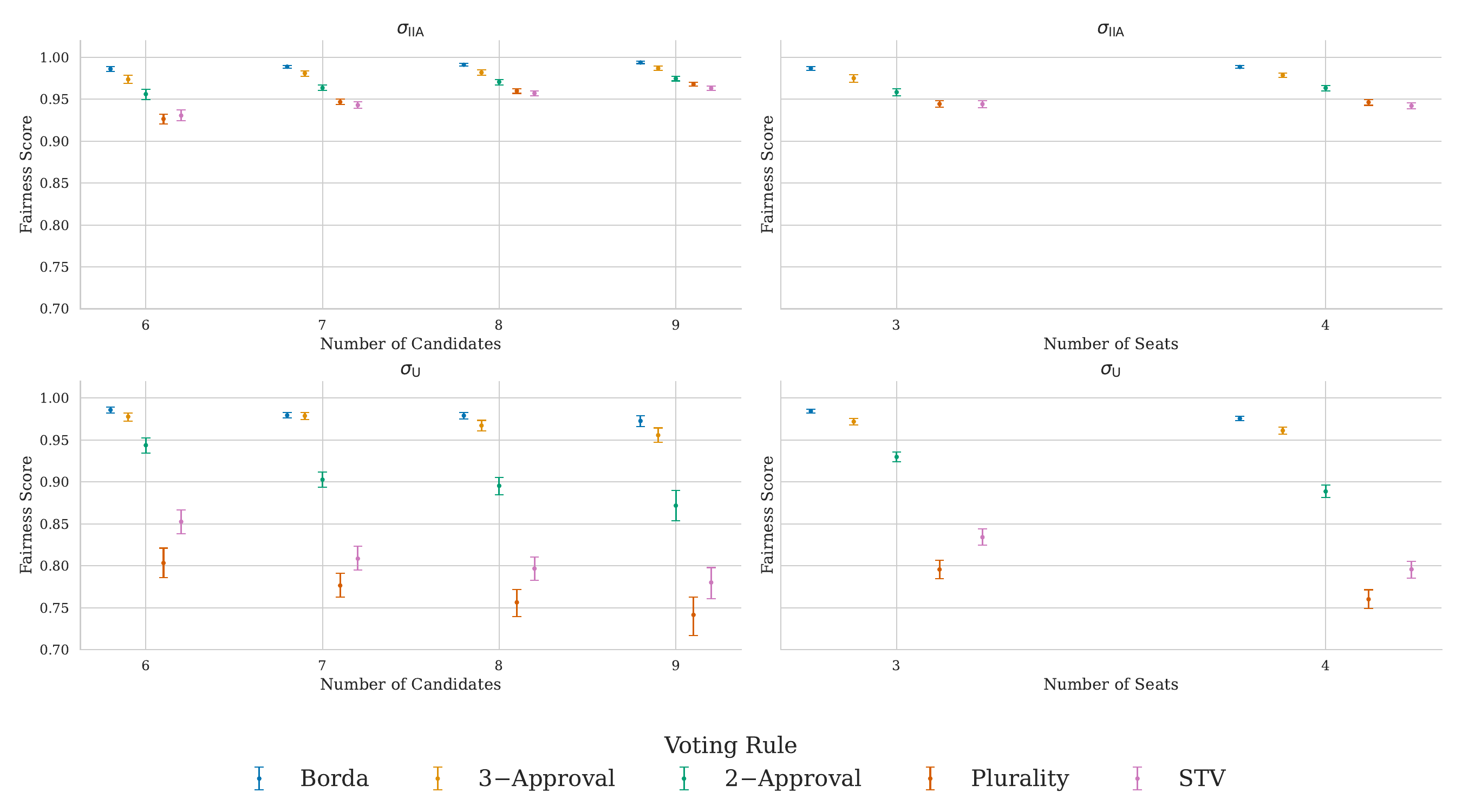}
 \caption{
 Bootstrapped mean and 95 \% confidence intervals for $\SIIA$ (top row) and $\SU$ metric (bottom row). 
}
\label{boot_fig}
\end{figure}

\section{Uncertainty Quantification}

Bootstrapping voter profiles are used to evaluate the reliability and variability of estimated fairness metrics for different voting rules, the number of candidates, and the number of seats. This approach helps to measure how sensitive the model results are to sampling variability in the original data (\citet{efron1994introduction}). The bootstrapping process involves repeatedly sampling the voters with replacement to create multiple resampled data sets (B = 1000). For each resample, the preferences of the key variables are recorded. The percentile bootstrap methodology produces a distribution of estimates from which confidence intervals, standard errors, and stability assessments can be calculated.  The inferential bootstrap results compliment the descriptive box plots of $ \SIIA$ and $ \SU$. 

Figure \ref{boot_fig} illustrates the notable consistency of the metric $\SIIA$ across the five scoring rules in terms of variability and monotonicity in the fairness score, considering both the number of candidates and the number of seats. In contrast, for $\SU$ , there are three groupings: \{Borda, 3-Approval\}, \{2-Approval\}, and \{Plurality, STV\}. However, within these groupings, the fairness scores remain consistent in varying the number of candidates and seats.

\begin{samepage}
\section{Additional plots}

\begin{figure}[H]
  \centering
  \includegraphics[width=0.99\linewidth]{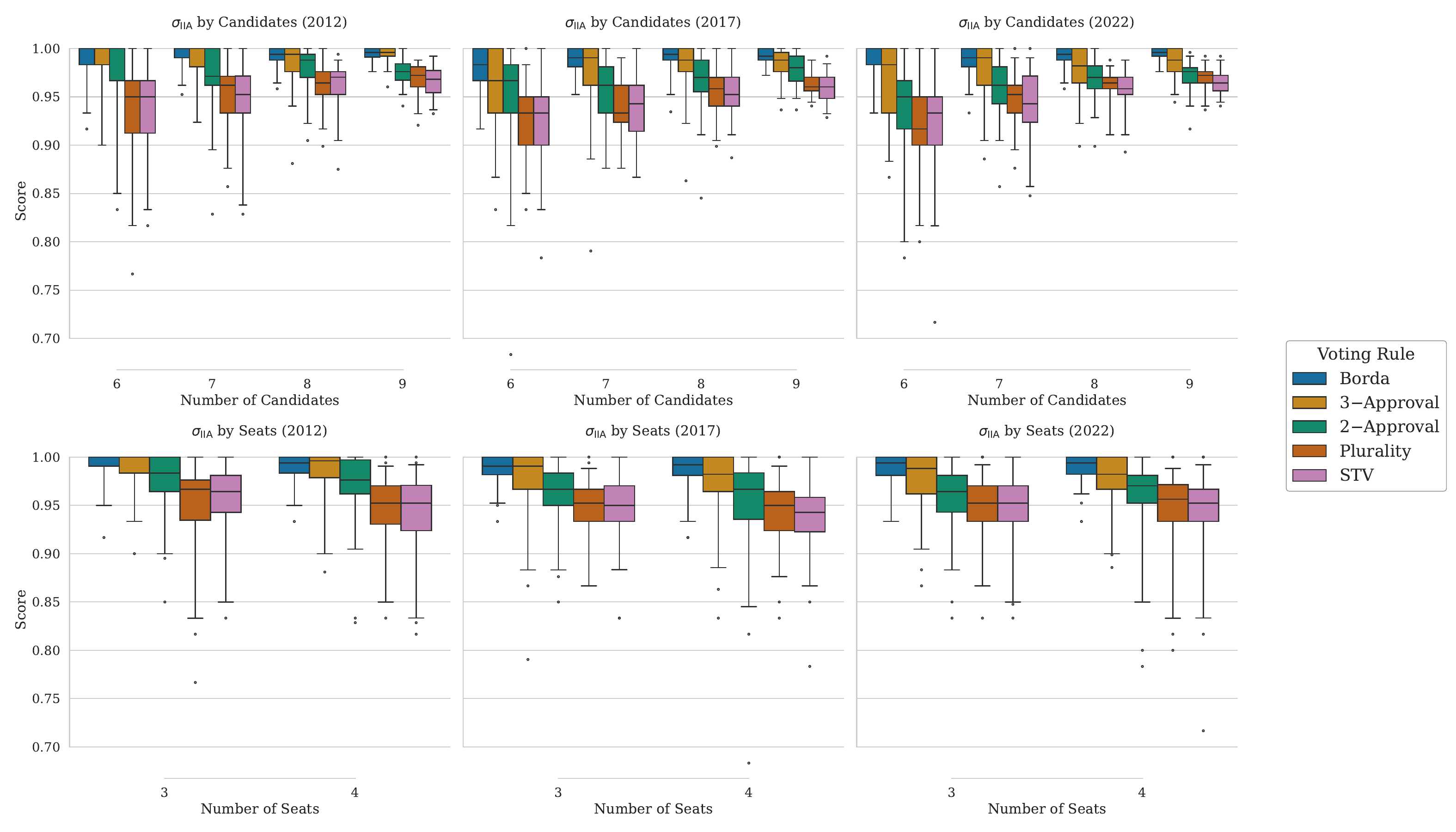}
 \caption{
$ \SIIA$ distributions by years (2012,2017,2022) of five voting rules on Scottish local elections with 6–9 candidates and 3-4 seats 
}

\end{figure}

\begin{figure}[H]
  \centering
  \includegraphics[width=0.99\linewidth]{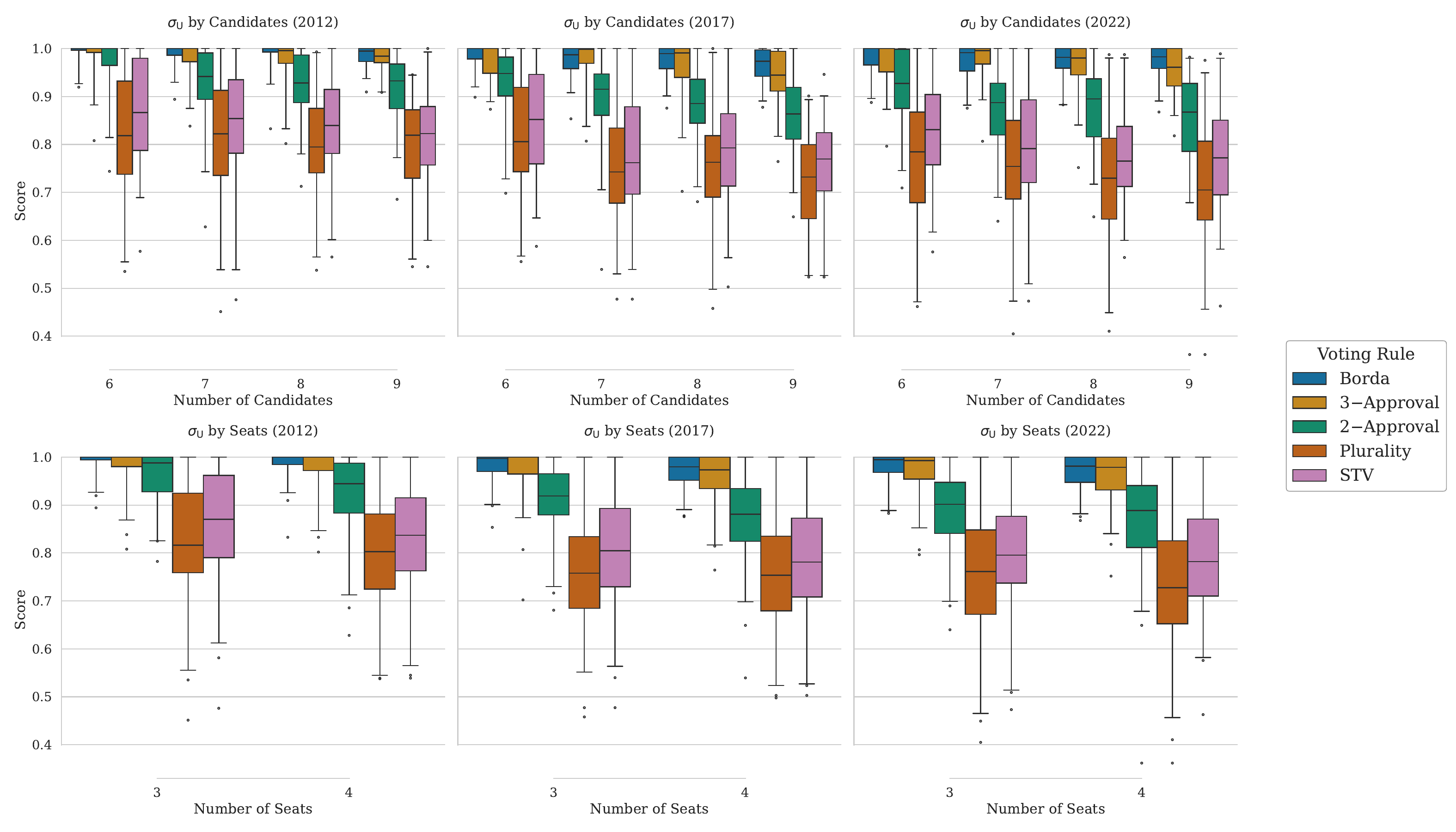}
 \caption{
$ \SU$ distributions by years (2012, 2017,2022) of five voting rules on Scottish local elections with 6–9 candidates and 3-4 seats 
}

\end{figure}
\end{samepage}

\section{Hardware, Software, and Reproducibility Details}

All experiments were conducted on a MacBook Pro with an Apple M1 chip (8-core CPU: 4 performance and 4 efficiency cores) and 16 GB of unified memory, running macOS 13.x. We used Python 3.10.1 along with the following Python libraries: \texttt{votekit 3.1.0}, \texttt{pandas}, \texttt{numpy}, \texttt{matplotlib}, \texttt{geopandas} and \texttt{seaborn}. The complete analysis of the Scottish STV election data took approximately 2 hours, while the Bradley--Terry simulation experiments required around 3 hours of compute time. All the data and code used in this paper are publicly available at \url{https://github.com/Suvadip2776/Quantitative_fairness}.


\end{document}